\documentclass[a4paper,10pt]{article}
\usepackage[utf8]{inputenc}
\usepackage{mathrsfs,pgf,tikz,amsmath,amssymb,hyperref,amsthm,stmaryrd}
\usepackage{colortbl}
\usetikzlibrary{arrows}

\theoremstyle{plain}
\newtheorem{theorem}{Theorem}[section]                                          
\newtheorem{proposition}[theorem]{Proposition}

\theoremstyle{definition}
\newtheorem{definition}[theorem]{Definition}
\theoremstyle{remark}
\newtheorem{remark}[theorem]{Remark}

\newtheorem{examples}[theorem]{Examples}

\DeclareMathOperator{\p}{P}

\title{Benford or not Benford: new results on digits beyond the first}

\begin{document}

\maketitle

\begin{abstract}
In this paper, we will see that the proportion of $d$ as $p^{\text{th}}$ digit, where $p>1$ and $d\in\llbracket0,9\rrbracket$, in data (obtained thanks 
to the hereunder developed model) is more likely to follow a law whose probability distribution is determined by a specific upper bound, rather than the 
generalization of Benford's Law to digits beyond the first one. These probability distributions fluctuate around theoretical values determined by Hill 
in $1995$. Knowing beforehand the value of the upper bound can be a way to find a better adjusted law than Hill's one. 
\end{abstract}

\section*{Introduction}

Benford's Law is really amazing: according to it, the first digit $d$, $d\in\llbracket1,9\rrbracket$, of 
numbers in many naturally occurring collections of data does not follow a discrete uniform distribution; it rather follows a logarithmic 
distribution. Having been discovered by Newcomb in $1881$ (\cite{new}), this law was definitively brought to light by Benford in 
$1938$ (\cite{ben}). He proposed the following probability distribution: the probability for $d$ to be the first digit of a number is equal to 
$\log(1+\frac{1}{d})$. Most of the empirical data, as physical data (Knuth in \cite{knu} or Burke and Kincanon in \cite{BK}), demographic and economic 
data (Nigrini and Wood in \cite{NW}) or genome data (Friar \textit{et al.} in \cite{FGP}), follow approximately Benford's Law. To such an extent that 
this law is used to detect possible frauds in lists of socio-economic data (\cite{var}) or in scientific publications (\cite{die}). 

In \cite{blo}, Blondeau Da Silva, building a rather relevant representative model, showed that, in this case, the proportion of each $d$ as leading 
digit, $d\in\llbracket0,9\rrbracket$, structurally fluctuates. It strengthens the fact that, concerning empirical data sets, this law often appears to be a good 
approximation of the reality, but no more than an approximation (\cite{del}). We can note that there also exist distributions known to disobey 
Benford's Law (\cite{rai} and \cite{bee}). 

Generalizing Benford's Law, Hill (\cite{hit}) extends the law to digits beyond the first one: the probability for $d$, $d\in\llbracket0,9\rrbracket$, to 
be the $p^{\text{th}}$ digit of a number is equal to $\sum_{j=10^{p-2}}^{10^{p-1}-1}\log(1+\frac{1}{10j+d})$.

Building a very similar model to that described in \cite{blo}, the naturally occurring data will be considered as the realizations of independant random 
variables following the hereinafter constraints: $(a)$ the data is strictly positive and is upper-bounded by an integer $n$, constraint which is often valid 
in data sets, the physical, biological and economical quantities being limited ; $(b)$ each random variable is considered to follow a discrete uniform 
distribution whereby the first strictly positive $p$-digits integers ($p>1$) are equally likely to occur ($i$ being uniformly randomly selected in 
$\llbracket10^{p-1},n\rrbracket$). This model relies on the fact that the random variables are not always the same. 

Through this article we will demonstrate that the predominance of $0$ over $1$ (and of $1$ over $2$, and so on), as $p^{\text{th}}$, ($p>1$) digit is 
all but surprising. Hill's probabilities became standard values that should exactly be followed by most of naturally occurring collections of data. 
However the reality is that the proportion of each $d$ as leading digit structurally fluctuates. There is not a single law but numerous distinct laws 
that we will hereafter examine.

\section{Notations and probability space}

Let $p$ and $d$ be two strictly positive integers such that $p>1$ and $d\in\llbracket0,9\rrbracket$. Let $m$ be a strictly positive integer such that 
$m\ge10^{p-1}$. Let $\mathcal{U}\{10^{p-1},m\}$ denote the discrete uniform distribution whereby integers between $10^{p-1}$ and $m$ are equally likely 
to be observed.

Let $n$ be a strictly positive integer such that $n\ge 10^{p-1}$. Let us consider the random experiment $\mathcal{E}_n$ of tossing two independent dice. 
The first one is a fair $(n+1-10^{p-1})$-sided die showing $n+1-10^{p-1}$ different numbers from $1$ to $n+1-10^{p-1}$. The number $i$ rolled on it 
defines the number of faces on the second die. It thus shows $i$ different numbers from $10^{p-1}$ to $i+10^{p-1}-1$. 

Let us now define the probability space $\Omega_n$ as follows: $\Omega_n=\{(i,j):i\in\llbracket1,n+1-10^{p-1}\rrbracket\text{ and }
j\in\llbracket10^{p-1},i+10^{p-1}-1\rrbracket\}$. Our probability measure is denoted by $\p$. 

Let us denote by $D_{(n,p)}$ the random variable from $\Omega_n$ to $\llbracket0,9\rrbracket$ that maps each element $\omega$ of $\Omega_n$ to the 
$p^{\text{th}}$ digit of the second component of $\omega$.

As our aim is to determine the probability that the $p^{\text{th}}$ digit of the integer obtained thanks to the second throw is $d$, it can be 
considered with no consequence on our results that we first select an integer $i$ equal to or less than $n$ among at least $p$-digits integers 
(following the $\mathcal{U}\{10^{p-1},n\}$ discrete uniform distribution); afterwards we select an other at least $p$-digits integer equal to or less 
than $i$ (following the $\mathcal{U}\{10^{p-1},i\}$ discrete uniform distribution).

\section{Proportion of \texorpdfstring{$d$}{}}

Through the below proposition, we will express the value of $\p(D_{(n,p)}=d)$ \textit{i.e.} the probability that the $p^{\text{th}}$ digit of our second 
throw in the random experiment $\mathcal{E}_n$ is $d$.

\begin{proposition}\label{proi}Let $k$ denote the integer such that:
\begin{center}
$k=\max\{i\in\mathbb N:10^{i+p}\le n\}\,$.
\end{center}
Let $l$ denote the positive integer such that:
\begin{center}
$l=\lfloor\frac{n-(10^{p-1}+d)10^{k+1}}{10^{k+2}}\rfloor+10^{p-2}\,$.
\end{center}
The value of $\p(D_{(n,p)}=d)$ is:
\begin{align*}
\frac{1}{n+1-10^{p-1}}\Big(&\sum_{i=0}^k\big(\sum_{j=10^{p-2}}^{10^{p-1}-1}\sum_{b=(10j+d)10^i}^{(10j+(d+1))10^i-1}\frac{b-((9j+d)10^i+10^{p-2}-1)}
{b+1-10^{p-1}}\\
+&\sum_{j=10^{p-2}-1}^{10^{p-1}-1}\sum_{a=\max(10^{p+i-1},(10j+(d+1))10^i)}^{\min(10^{p+i}-1,(10(j+1)+d)10^i-1)}\frac{10^i(j+1)-10^{p-2}}{a+1-10^{p-1}}
\big)\\
+&\, r_{(n,d,p)}\Big)\,,
\end{align*}
where $r_{(n,d,p)}$ is, if the $p^{\text{th}}$ digit of $n$ is $d$:
\begin{align*}
&\sum_{j=10^{p-2}}^{l}\sum_{b=(10j+d)10^{k+1}}^{\min(n,(10j+(d+1))10^{k+1}-1)}\frac{b-((9j+d)10^{k+1}+10^{p-2}-1)}
{b+1-10^{p-1}}\\
+&\sum_{j=10^{p-2}-1}^{l-1}\sum_{a=\max(10^{p+k},(10j+(d+1))10^{k+1})}^{(10(j+1)+d)10^{k+1}-1}\frac{10^{k+1}(j+1)-10^{p-2}}{a+1-10^{p-1}}\,,
\end{align*}
or where $r_{(n,d,p)}$ is, if the $p^{\text{th}}$ digit of $n$ is all but $d$:
\begin{align*}
&\sum_{j=10^{p-2}}^{l}\sum_{b=(10j+d)10^{k+1}}^{(10j+(d+1))10^{k+1}-1}\frac{b-((9j+d)10^{k+1}+10^{p-2}-1)}{b+1-10^{p-1}}\\
+&\sum_{j=10^{p-2}-1}^{l}\sum_{a=\max(10^{p+k},(10j+(d+1))10^{k+1})}^{\min(n,(10(j+1)+d)10^{k+1}-1)}\frac{10^{k+1}(j+1)-10^{p-2}}{a+1-10^{p-1}}\,.
\end{align*}
\end{proposition}

\begin{proof}Let us denote by $F_{(n,p)}$ the random variable from $\Omega_n$ to $\llbracket1,n+1-10^{p-1}\rrbracket$ that maps each element $\omega$ of 
$\Omega_n$ to the first component of $\omega$. It returns the number obtained on the first throw of the unbiased $(n+1-10^{p-1})$-sided die. For each 
$q\in\llbracket1,n+1-10^{p-1}\rrbracket$, we have:
\begin{align}\label{equi}
\p(F_{(n,p)}=q)=\frac{1}{n+1-10^{p-1}}\,.
\end{align}

According to the Law of total probability, we state:
\begin{align}\label{espi}
\p(D_{(n,p)}=d)=\sum_{q=1}^{n+1-10^{p-1}}\p(D_{(n,p)}=d|F_{(n,p)}=q)\p(F_{(n,p)}=q)\,.
\end{align}

Thereupon two cases appear in determining the value, for $q\in\llbracket1,n+1-10^{p-1}\rrbracket$, of $\p(D_{(n,p)}=d|F_{(n,p)}=q)$. Let $k_q$ be the 
integer such that $k_q=\max\{k\in\mathbb N:10^{p+k}\le q+10^{p-1}-1\}$ in both cases. 

Let us study the first case where the $p^{\text{th}}$ digit of $q+10^{p-1}-1$ is $d$. For all $i$ in $\llbracket0,k_q\rrbracket$, there exist 
$9\times10^{p-2}$ sequences of $10^i$ consecutive integers from $(10j+d)10^i$ to $(10j+(d+1))10^i-1$, where $j\in\llbracket10^{p-2},10^{p-1}-1
\rrbracket$, whose $p^{\text{th}}$ digit is $d$. The higher of these integers is $(10(10^{p-1}-1)+(d+1))10^{k_q}-1$, the last $(p+k_q)$-digit number 
in this case. Thus, from $10^{p-1}$ to $10^{p+k_q}-1$, the number of integers whose $p^{\text{th}}$ digit is $d$ is:
\begin{align*}
\sum_{i=0}^{k_q}\sum_{j=10^{p-2}}^{10^{p-1}-1}\sum_{(10j+d)10^i}^{(10j+(d+1))10^i-1}1=9\times10^{p-2}\sum_{i=0}^{k_q}10^i=10^{p-2}(10^{k_q+1}-1)\,.
\end{align*}
This equality still holds true for $k_q=-1$. Such types of sum would be considered null in the rest of the article. From $10^{p+k_q}$ to $q+10^{p-1}-1$, 
there exist $t$ sequences of $10^{k_q+1}$ consecutive integers from $(10j+d)10^{k_q+1}$ to $(10j+(d+1))10^{k_q+1}-1$, where $j\in\llbracket10^{p-2},
10^{p-2}+t-1\rrbracket$, whose $p^{\text{th}}$ digit is $d$. There also exist $q+10^{p-1}-1-(10(10^{p-2}+t)+d)10^{k_q+1}+1$ additional integers in this 
case between $(10(10^{p-2}+t)+d)10^{k_q+1}$ and $q+10^{p-1}-1$. Finally the total number of integers whose $p^{\text{th}}$ digit is $d$ is:
\begin{center}
$10^{p-2}(10^{k_q+1}-1)+t\times10^{k_q+1}+q+10^{p-1}-1-(10(10^{p-2}+t)+d)10^{k_q+1}+1$\\
\text{i.e.}\quad $q+10^{p-1}-1-\Big(\big(9(10^{p-2}+t)+d\big)10^{k_q+1}-1\Big)\,$.
\end{center}
It may be inferred that:
\begin{equation}\label{p1i}
\p(D_{(n,p)}=d|F_{(n,p)}=q)=\frac{q+10^{p-1}-1-\Big(\big(9(10^{p-2}+t)+d\big)10^{k_q+1}-1\Big)}{q}\,,
\end{equation}
the $p^{\text{th}}$ digit of $q+10^{p-1}-1$ being here $d$.

In the second case, we consider the integers $q+10^{p-1}-1$ whose $p^{\text{th}}$ digits are different from $d$. On the basis of the previous case, the 
total number of integers whose $p^{\text{th}}$ digit is $d$ is, where $t$ is the number of sequences of consecutive integers lower than $q+10^{p-1}-1$:
\begin{center}
$10^{p-2}(10^{k_q+1}-1)+t\times10^{k_q+1}$\\
\text{i.e.}\quad $10^{k_q+1}(10^{p-2}+t)-10^{p-2}\,$.
\end{center}
It can be concluded that:
\begin{equation}\label{p2i}
\p(D_{(n,p)}=d|F_{(n,p)}=q)=\frac{10^{k_q+1}(10^{p-2}+t)-10^{p-2}}{q}\,,
\end{equation}
the $p^{\text{th}}$ digit of $q+10^{p-1}-1$ being here different from $d$.\\
Using equalities (\ref{equi}), (\ref{espi}), (\ref{p1i}) and (\ref{p2i}), we get our result.
\end{proof}

For example, we get:
\begin{examples}
Let us first determine the value of $\p(D_{(10003,5)}=2)$. The probability that the fifth digit of a randomly selected number in 
$\llbracket10000,10000\rrbracket$ is $2$ is $\frac{0}{1}$, those in $\llbracket10000,10001\rrbracket$ is $\frac{0}{2}$, those in 
$\llbracket10000,10002\rrbracket$ is $\frac{1}{3}$ and those in $\llbracket10000,10003\rrbracket$ is $\frac{1}{4}$. Hence we have:
\begin{align*}
\p(D_{(10003,5)}=2)=\frac{1}{4}\Big(\frac{0}{1}+\frac{0}{2}+\frac{1}{3}+\frac{1}{4}\Big)\approx0.1458\,.
\end{align*}
It is the second case of Proposition \ref{proi}, where $n=10003$, $d=2$, $p=5$, $k=-1$ and $l=1000$.\\
Let us now determine the value of $\p(D_{(1113,3)}=1)$ (first case of Proposition \ref{proi}); in this case, we have $k=0$ and $l=11$.
\begingroup\scriptsize
\begin{align*}
\p(D_{(1113,3)}=1)&=\frac{1}{1014}\Big(\sum_{j=10}^{99}\frac{j-9}{10j-98}+\sum_{j=10}^{98}\sum_{a=10j+2}^{10(j+1)}\frac{j-9}{a-99}+
\sum_{a=992}^{999}\frac{90}{a-99}+\sum_{a=1000}^{1009}\frac{90}{a-99}\\
&+\sum_{b=1010}^{1019}\frac{b-919}{b-99}+\sum_{a=1020}^{1109}\frac{100}{a-99}+\sum_{b=1110}^{1113}\frac{b-1009}{b-99}\Big)\\
&=\frac{1}{1014}\Big(\frac{1}{2}+\frac{1}{3}+\frac{1}{4}+...+\frac{1}{11}+\frac{2}{12}+\frac{2}{13}+...+\frac{89}{891}+\frac{90}{892}+\frac{90}{893}+
...+\frac{90}{910}\\
&+\frac{91}{911}+...+\frac{100}{920}+\frac{100}{921}+...+\frac{100}{1010}+\frac{101}{1011}+...+\frac{104}{1014}\Big)\\
&\approx0.1028\,.
\end{align*}
\endgroup
Let us determine the value of $\p(D_{(212,2)}=9)$ (second case of Proposition \ref{proi}); in this case, we have $k=0$ and $l=1$.
\begingroup\scriptsize
\begin{align*}
\p(D_{(212,2)}=9)&=\frac{1}{203}\Big(\frac{9}{10}+\sum_{j=1}^{8}\sum_{a=10(j+1)}^{10(j+1)+8}\frac{j}{a-9}+
\sum_{a=100}^{189}\frac{9}{a-9}+\sum_{b=190}^{199}\frac{b-180}{b-9}+\sum_{a=200}^{212}\frac{19}{a-9}\Big)\\
&=\frac{1}{203}\Big(\frac{1}{10}+\frac{1}{11}+...+\frac{1}{19}+\frac{2}{20}+\frac{2}{21}+...+\frac{8}{89}+\frac{9}{90}+\frac{9}{91}+...+\frac{9}{180}+
\frac{10}{181}\\
&+\frac{11}{182}+...+\frac{19}{190}+\frac{19}{191}+...+\frac{19}{203}\Big)\\
&\approx0.0759\,.
\end{align*}
\endgroup
\end{examples}

\section{Study of a particular subsequence}

It is natural that we take a specific look at the values of $n$ positioned one rank before the integers for which the number of digits has just increased.

To this end we will consider the sequence $\big(\p(D_{n,p}=d)\big)_{n\in\mathbb N\setminus\llbracket0,10^{p-1}-1\rrbracket}$. In the interests of 
simplifying notation, we will denote by $(P_{(d,n,p)})_{n\in\mathbb N\setminus\llbracket0,10^{p-1}-1\rrbracket}$ this sequence. Let us study the 
subsequence $(P_{(d,\phi_{(d,p)}(n),p)})_{n\in\mathbb N\setminus\llbracket0,p-1\rrbracket}$ where $\phi_{(d,p)}$ is the function from 
$\mathbb N\setminus\llbracket0,p-1\rrbracket$ to $\mathbb N$ that maps $n$ to $10^n-1$. We get the below result:
\begin{proposition}\label{sub1i}
The subsequence $(P_{(d,\phi_{(d,p)}(n),p)})_{n\in\mathbb N\setminus\llbracket0,p-1\rrbracket}$ converges to:
\begingroup\scriptsize
\begin{align*}
10^{-1}+\frac{n_{(d,p)}+m_{(d,p)}-9l_{(d,p)}-d\times k_{(d,p)}}{9\times10^{p-1}}+\frac{1}{90}\ln(\frac{10^{p-1}+d}{10^{p-1}})
+\frac{1}{9}\ln(\frac{10^{p}}{10^{p}-10+d+1})\,,
\end{align*}
\endgroup
where:
\[\left \{
\begin{array}{l @{=} l}
    k_{(d,p)} & \sum_{j=10^{p-2}}^{10^{p-1}-1}\ln(\frac{10j+(d+1)}{10j+d}) \\
    l_{(d,p)} & \sum_{j=10^{p-2}}^{10^{p-1}-1}j\ln(\frac{10j+(d+1)}{10j+d}) \\
    m_{(d,p)} & \sum_{j=10^{p-2}}^{10^{p-1}-2}\ln(\frac{10(j+1)+d}{10j+(d+1)}) \\
    n_{(d,p)} & \sum_{j=10^{p-2}}^{10^{p-1}-2}j\ln(\frac{10(j+1)+d}{10j+(d+1)})\,.
\end{array}
\right.\]
\end{proposition}

\begin{proof}
Let $n$ be a positive integer such that $n\ge p$. According to Proposition \ref{proi}, we have $P_{(d,\phi_{(d,p)}(n),p)}=P_{(d,10^n-1,p)}$ 
\textit{i.e.}, knowing that in this case $k=\max\{i\in\mathbb N:10^{i+p}\le 10^n-1\}=n-p-1$:
\begin{align*}
\frac{1}{10^n-10^{p-1}}\Big(&\sum_{i=0}^{n-p}\big(\sum_{j=10^{p-2}}^{10^{p-1}-1}\sum_{b=(10j+d)10^i}^{(10j+(d+1))10^i-1}\frac{b-((9j+d)10^i+10^{p-2}-1)}
{b+1-10^{p-1}}\\
+&\sum_{j=10^{p-2}-1}^{10^{p-1}-1}\sum_{a=\max(10^{p+i-1},(10j+(d+1))10^i)}^{\min(10^{p+i}-1,(10(j+1)+d)10^i-1)}\frac{10^i(j+1)-10^{p-2}}{a+1-10^{p-1}}
\big)\Big)\,.
\end{align*}
Let us denote by $b_{(i,d,p)}$ the positive number:
\begin{align*}
\sum_{j=10^{p-2}}^{10^{p-1}-1}\sum_{b=(10j+d)10^i}^{(10j+(d+1))10^i-1}\frac{b-((9j+d)10^i+10^{p-2}-1)}{b+1-10^{p-1}}\,,
\end{align*}
and by $a_{(i,d,p)}$ the positive number: 
\begin{align*}
\sum_{j=10^{p-2}-1}^{10^{p-1}-1}\sum_{a=\max(10^{p+i-1},(10j+(d+1))10^i)}^{\min(10^{p+i}-1,(10(j+1)+d)10^i-1)}\frac{10^i(j+1)-10^{p-2}}{a+1-10^{p-1}}\,.
\end{align*}
Thus we have:
\begin{align*}
P_{(d,\phi_{(d,p)}(n),p)}=\frac{1}{10^n-10^{p-1}}&\sum_{i=0}^{n-p}\Big(b_{(i,d,p)}+a_{(i,d,p)}\Big)\,.
\end{align*}

Let us first find an appropriate lower bound of $P_{(d,\phi_{(d,p)}(n),p)}$. We have:
\begingroup\scriptsize
\begin{align*}
b_{(i,d,p)}&=\sum_{j=10^{p-2}}^{10^{p-1}-1}\big(10^i-\sum_{b=(10j+d)10^i}^{(10j+(d+1))10^i-1}\frac{(9j+d)10^i+10^{p-2}-10^{p-1}}{b+1-10^{p-1}}\big)\\
&=9\times10^{p+i-2}-\sum_{j=10^{p-2}}^{10^{p-1}-1}((9j+d)10^i+10^{p-2}-10^{p-1})\sum_{b=(10j+d)10^i}^{(10j+(d+1))10^i-1}\frac{1}{b+1-10^{p-1}}
\end{align*}
\endgroup
Recall that for all integers $(p,q)$, such that $1<p<q$: 
\begin{align}\label{inelog}
\ln(\frac{q+1}{p})\le\sum_{k=p}^q\frac{1}{k}\le\ln(\frac{q}{p-1})\,.
\end{align}
Consequently, we obtain, for $i\ge1$:
\begingroup\scriptsize
\begin{align*}
b_{(i,d,p)}&\ge9\times10^{p+i-2}-\sum_{j=10^{p-2}}^{10^{p-1}-1}(9j+d)10^i\ln(\frac{(10j+(d+1))10^i-10^{p-1}}{(10j+d)10^i-10^{p-1}})\\
&\ge9\times10^{p+i-2}-\sum_{j=10^{p-2}}^{10^{p-1}-1}(9j+d)10^i\big(\ln(\frac{10j+(d+1)}{10j+d})+\ln(1+\frac{\frac{10^{p-1}}{10j+(d+1)}}{10^i(10j+d)-
10^{p-1}})\big)\\
&\ge9\times10^{p+i-2}-d\times10^i\sum_{j=10^{p-2}}^{10^{p-1}-1}\ln(\frac{10j+(d+1)}{10j+d})-9\times10^i\sum_{j=10^{p-2}}^{10^{p-1}-1}j\ln(\frac{10j+(d+1)}{10j+d})\\
&\enspace-\sum_{j=10^{p-2}}^{10^{p-1}-1}(9j+d)10^i\ln(1+\frac{\frac{10^{p-1}}{10j+(d+1)}}{10^i(10j+d)-10^{p-1}})\big)\,.
\end{align*}
\endgroup
Let us denote by $k_{(d,p)}$ the positive number $\sum_{j=10^{p-2}}^{10^{p-1}-1}\ln(\frac{10j+(d+1)}{10j+d})$ and $l_{(d,p)}$ the positive number 
$\sum_{j=10^{p-2}}^{10^{p-1}-1}j\ln(\frac{10j+(d+1)}{10j+d})$. Knowing that for all $x\in]-1;+\infty[$, we have $\ln(1+x)\le x$, we obtain:
\begingroup\scriptsize
\begin{align*}
b_{(i,d,p)}&\ge9\times10^{p+i-2}-d\times10^ik_{(d,p)}-9\times10^il_{(d,p)}-\sum_{j=10^{p-2}}^{10^{p-1}-1}(9j+d)10^i\frac{\frac{10^{p-1}}{10j+(d+1)}}{10^i(10j+d)-10^{p-1}}\\
&\ge9\times10^{p+i-2}-d\times10^ik_{(d,p)}-9\times10^il_{(d,p)}-\sum_{j=10^{p-2}}^{10^{p-1}-1}10^i\frac{10^{p-1}}{10^i\times10^{p-1}-10^{p-1}}\\
&\ge9\times10^{p+i-2}-d\times10^ik_{(d,p)}-9\times10^il_{(d,p)}-9\times10^{p-2}\frac{10^i}{10^i-1}\,.
\end{align*}
\endgroup
Similarly, we have thanks to inequalities (\ref{inelog}):
\begingroup\scriptsize
\begin{align*}
a_{(i,d,p)}&\ge\sum_{j=10^{p-2}}^{10^{p-1}-2}(10^i(j+1)-10^{p-2})\ln(\frac{(10(j+1)+d)10^i+1-10^{p-1}}{(10j+(d+1))10^i+1-10^{p-1}})\\
&\enspace+(10^{p-2+i}-10^{p-2})\ln(\frac{(10^{p-1}+d)10^i+1-10^{p-1}}{10^{p+i-1}+1-10^{p-1}})\\
&\enspace+(10^{p-1+i}-10^{p-2})\ln(\frac{10^{p+i}+1-10^{p-1}}{(10^{p}-10+d+1)10^i+1-10^{p-1}})\\
&\ge10^i\sum_{j=10^{p-2}}^{10^{p-1}-2}j\ln(\frac{10(j+1)+d}{10j+(d+1)})+10^i\sum_{j=10^{p-2}}^{10^{p-1}-2}j\ln(1+\frac{\frac{9\times(10^{p-1}-1)}
{10(j+1)+d}}{(10j+d+1)10^i+1-10^{p-1}})\\
&\enspace+(10^i-10^{p-2})\Big(\sum_{j=10^{p-2}}^{10^{p-1}-2}\big(\ln(\frac{10(j+1)+d}{10j+(d+1)})+
\ln(1+\frac{\frac{9\times(10^{p-1}-1)}{10(j+1)+d}}{(10j+d+1)10^i+1-10^{p-1}})\big)\Big)\\
&\enspace+(10^{p-2+i}-10^{p-2})\big(\ln(\frac{10^{p-1}+d}{10^{p-1}})+\ln(1+\frac{\frac{d(10^{p-1}-1)}{10^{p-1}+d}}{10^{p-1}10^i+1-10^{p-1}})\big)\\
&\enspace+(10^{p-1+i}-10^{p-2})(\ln(\frac{10^{p}}{10^{p}-10+d+1})+\ln(1+\frac{\frac{(10^{p-1}-1)(10-d-1)}{10^p}}{(10^p-10+d+1)10^i+1-10^{p-1}})).
\end{align*}
\endgroup
Let us denote by $m_{(d,p)}$ the positive number $\sum_{j=10^{p-2}}^{10^{p-1}-2}\ln(\frac{10(j+1)+d}{10j+(d+1)})$ and $n_{(d,p)}$ the positive 
number $\sum_{j=10^{p-2}}^{10^{p-1}-2}j\ln(\frac{10(j+1)+d}{10j+(d+1)})$:
\begingroup\scriptsize
\begin{align*}
a_{(i,d,p)}&\ge10^in_{(d,p)}+(10^i-10^{p-2})\Big(m_{(d,p)}+\sum_{j=10^{p-2}}^{10^{p-1}-2}
\ln(1+\frac{\frac{9\times(10^{p-1}-1)}{10(j+1)+d}}{(10j+d+1)10^i+1-10^{p-1}})\Big)\\
&\enspace+(10^{p-2+i}-10^{p-2})\ln(\frac{10^{p-1}+d}{10^{p-1}})+(10^{p-1+i}-10^{p-2})\ln(\frac{10^{p}}{10^{p}-10+d+1})\,.
\end{align*}
\endgroup
Hence we have:
\begingroup\scriptsize
\begin{align*}
P_{(d,\phi_{(d,p)}(n),p)}&\ge\frac{1}{10^n}\Big(a_{(0,d,p)}+b_{(0,d,p)}+\sum_{i=1}^{n-p}\big(9\times10^{p+i-2}-d\times10^ik_{(d,p)}
-9\times10^il_{(d,p)}\\
&\enspace+10^in_{(d,p)}+10^im_{(d,p)}+10^{p-2+i}\ln(\frac{10^{p-1}+d}{10^{p-1}})+10^{p-1+i}\ln(\frac{10^{p}}{10^{p}-10+d+1})\\
&\enspace-9\times10^{p-2}\frac{10}{9}-10^{p-2}\big(m_{(d,p)}+\ln(\frac{10^{p-1}+d}{10^{p-1}})+\ln(\frac{10^{p}}{10^{p}-10+d+1})\big)\\
&\enspace+(10^i-10^{p-2})\sum_{j=10^{p-2}}^{10^{p-1}-2}\ln(1+\frac{\frac{9\times(10^{p-1}-1)}{10(j+1)+d}}{(10j+d+1)10^i+1-10^{p-1}})\big)\Big)\,.
\end{align*}
\endgroup
In light of the following equality $\sum_{i=1}^{n-p}10^i=\frac{10^{n-p+1}-10}{9}$, we have:
\begingroup\scriptsize
\begin{align*}
P_{(d,\phi_{(d,p)}(n),p)}&\ge10^{-1}+\frac{10^{-p+1}(n_{(d,p)}+m_{(d,p)}-9l_{(d,p)}-dk_{(d,p)})}{9}+\frac{10^{-1}}{9}\ln(\frac{10^{p-1}+d}{10^{p-1}})\\
&\enspace+\frac{1}{9}\ln(\frac{10^{p}}{10^{p}-10+d+1})+\epsilon_{(d,n,p)}\,,
\end{align*}
\endgroup
where $\epsilon_{(d,n,p)}$ is:
\begingroup\scriptsize
\begin{align*}
&\frac{a_{(0,d,p)}+b_{(0,d,p)}}{10^n}-\frac{10^{p-1}}{10^n}+\frac{dk_{(d,p)}+9l_{(d,p)}-n_{(d,p)}-m_{(d,p)}}{9\times10^{n-1}}
-\frac{10^{p-1}}{9\times10^n}\ln(\frac{10^{p-1}+d}{10^{p-1}})\\
&-\frac{10^p}{9\times10^n}\ln(\frac{10^{p}}{10^{p}-10+d+1})-\frac{10^{p-1}(n-p)}{10^n}-\frac{10^{p-2}(n-p)}{10^n}\big(m_{(d,p)}+
\ln(\frac{10^{p-1}+d}{10^{p-1}})\\
&\enspace+\ln(\frac{10^{p}}{10^{p}-10+d+1})\big)+\frac{1}{10^n}\sum_{i=1}^{p-3}(10^i-10^{p-2})\sum_{j=10^{p-2}}^{10^{p-1}-2}
\ln(1+\frac{\frac{9\times(10^{p-1}-1)}{10(j+1)+d}}{(10j+d+1)10^i+1-10^{p-1}})\,.
\end{align*}
\endgroup
Knowing that for all $x\in]-1;+\infty[$, we have $\ln(1+x)\le x$, we obtain, for all $i\in\{1,...,p-3\}$:
\begingroup\scriptsize
\begin{align*}
\sum_{j=10^{p-2}}^{10^{p-1}-2}\ln(1+\frac{\frac{9\times(10^{p-1}-1)}{10(j+1)+d}}{(10j+d+1)10^i+1-10^{p-1}})
&\le\sum_{j=10^{p-2}}^{10^{p-1}-2}\frac{\frac{9\times(10^{p-1}-1)}{10(j+1)+d}}{(10j+d+1)10^i+1-10^{p-1}}\\
&\le10^{p-1}\frac{\frac{10^p}{10^{p-1}}}{d+2}\le10^p
\end{align*}
\endgroup
From the above upper bound and the definition of $\epsilon_{(d,n,p)}$, it may be deduced that $\lim\limits_{n\to+\infty}\epsilon_{(d,n,p)}=0$.

Let us now find an appropriate upper bound of $P_{(d,\phi_{(d,p)}(n),p)}$. Thanks to inequalities (\ref{inelog}):
\begingroup\scriptsize
\begin{align*}
b_{(i,d,p)}&\le9\times10^{p+i-2}-\sum_{j=10^{p-2}}^{10^{p-1}-1}((9j+d)10^i+10^{p-2}-10^{p-1})\\
&\enspace\ln(\frac{(10j+(d+1))10^i+1-10^{p-1}}{(10j+d)10^i+1-10^{p-1}})\\
&\le9\times10^{p+i-2}-\sum_{j=10^{p-2}}^{10^{p-1}-1}((9j+d)10^i+10^{p-2}-10^{p-1})\big(\ln(\frac{10j+(d+1)}{10j+d})\\
&\enspace+\ln(1+\frac{\frac{10^{p-1}-1}{10j+(d+1)}}{10^i(10j+d)+1-10^{p-1}})\big)\\
&\le9\times10^{p+i-2}-d\times10^ik_{(d,p)}-9\times10^il_{(d,p)}+10^{p-1}k_{(d,p)}\,.
\end{align*}
\endgroup
Similarly, we have thanks to inequalities (\ref{inelog}):
\begingroup\scriptsize
\begin{align*}
a_{(i,d,p)}&\le\sum_{j=10^{p-2}}^{10^{p-1}-2}10^i(j+1)\ln(\frac{(10(j+1)+d)10^i-10^{p-1}}{(10j+(d+1))10^i-10^{p-1}})\\
&\enspace+10^{p-2+i}\ln(\frac{(10^{p-1}+d)10^i-10^{p-1}}{10^{p+i-1}-10^{p-1}})+10^{p-1+i}\ln(\frac{10^{p+i}-10^{p-1}}{(10^{p}-10+d+1)10^i-10^{p-1}})\\
&\le10^in_{(d,p)}+10^i\sum_{j=10^{p-2}}^{10^{p-1}-2}j\ln(1+\frac{\frac{9\times10^{p-1}}{10(j+1)+d}}{(10j+d+1)10^i-10^{p-1}})\\
&\enspace+10^i\Big(m_{(d,p)}+\sum_{j=10^{p-2}}^{10^{p-1}-2}\ln(1+\frac{\frac{9\times10^{p-1}}{10(j+1)+d}}{(10j+d+1)10^i-10^{p-1}})\Big)\\
&\enspace+10^{p-2+i}\big(\ln(\frac{10^{p-1}+d}{10^{p-1}})+\ln(1+\frac{\frac{d\times10^{p-1}}{10^{p-1}+d}}{10^{p-1}10^i-10^{p-1}})\big)\\
&\enspace+10^{p-1+i}\big(\ln(\frac{10^{p}}{10^{p}-10+d+1})+\ln(1+\frac{\frac{10^{p-1}(10-d-1)}{10^p}}{(10^p-10+d+1)10^i-10^{p-1}})\big).
\end{align*}
\endgroup
Hence we have:
\begingroup\scriptsize
\begin{align*}
P_{(d,\phi_{(d,p)}(n),p)}&\le\frac{1}{10^n-10^{p-1}}\sum_{i=0}^{n-p}\Big(9\times10^{p+i-2}-d\times10^ik_{(d,p)}-9\times10^il_{(d,p)}+10^im_{(d,p)}\\
&\enspace+10^in_{(d,p)}+10^{p-2+i}\ln(\frac{10^{p-1}+d}{10^{p-1}})+10^{p-1+i}\ln(\frac{10^{p}}{10^{p}-10+d+1})\\
&\enspace+10^{p-1}k_{(d,p)}+10^i\sum_{j=10^{p-2}}^{10^{p-1}-2}j\ln(1+\frac{\frac{9\times10^{p-1}}{10(j+1)+d}}{(10j+d+1)10^i-10^{p-1}})\\
&\enspace+10^i\sum_{j=10^{p-2}}^{10^{p-1}-2}\ln(1+\frac{\frac{9\times10^{p-1}}{10(j+1)+d}}{(10j+d+1)10^i-10^{p-1}})\\
&\enspace+10^{p-2+i}\ln(1+\frac{\frac{d\times10^{p-1}}{10^{p-1}+d}}{10^{p-1}10^i-10^{p-1}})\\
&\enspace+10^{p-1+i}\ln(1+\frac{\frac{10^{p-1}(10-d-1)}{10^p}}{(10^p-10+d+1)10^i-10^{p-1}})\Big)\,.
\end{align*}
\endgroup
In light of the following equality $\sum_{i=0}^{n-p}10^i=\frac{10^{n-p+1}-1}{9}$, we have:
\begingroup\scriptsize
\begin{align*}
\lim\limits_{n\to+\infty}(\frac{1}{10^n-10^{p-1}}\sum_{i=0}^{n-p}9\times10^{p+i-2})&=10^{-1}\\
\lim\limits_{n\to+\infty}(-\frac{1}{10^n-10^{p-1}}\sum_{i=0}^{n-p}d\times10^ik_{(d,p)})&=\frac{-dk_{(d,p)}}{9\times10^{p-1}}\\
\lim\limits_{n\to+\infty}(-\frac{1}{10^n-10^{p-1}}\sum_{i=0}^{n-p}9\times10^il_{(d,p)})&=-l_{(d,p)}10^{1-p}\\
\lim\limits_{n\to+\infty}(\frac{1}{10^n-10^{p-1}}\sum_{i=0}^{n-p}10^in_{(d,p)})&=\frac{m_{(d,p)}}{9\times10^{p-1}}\\
\lim\limits_{n\to+\infty}(\frac{1}{10^n-10^{p-1}}\sum_{i=0}^{n-p}10^in_{(d,p)})&=\frac{n_{(d,p)}}{9\times10^{p-1}}\\
\lim\limits_{n\to+\infty}(\frac{1}{10^n-10^{p-1}}\sum_{i=0}^{n-p}10^{p-2+i}\ln(\frac{10^{p-1}+d}{10^{p-1}}))&=
\frac{1}{90}\ln(\frac{10^{p-1}+d}{10^{p-1}})\\
\lim\limits_{n\to+\infty}(\frac{1}{10^n-10^{p-1}}\sum_{i=0}^{n-p}10^{p-1+i}\ln(\frac{10^{p}}{10^{p}-10+d+1}))&=
\frac{1}{9}\ln(\frac{10^{p}}{10^{p}-10+d+1}))\\
\lim\limits_{n\to+\infty}(\frac{1}{10^n-10^{p-1}}\sum_{i=0}^{n-p}10^{p-1}k_{(d,p)})&=0\,.
\end{align*}
\endgroup
Knowing that for all $x\in]-1;+\infty[$, we have $\ln(1+x)\le x$, we obtain, for $i\ge1$:
\begingroup\scriptsize
\begin{align*}
&10^i\sum_{j=10^{p-2}}^{10^{p-1}-2}j\ln(1+\frac{\frac{9\times10^{p-1}}{10(j+1)+d}}{(10j+d+1)10^i-10^{p-1}})\le10^{i+p-1}
\frac{\frac{10^{p}}{10^{p-1}}}{10^{p-1}10^i-10^{p-1}}=\frac{10^{i+1}}{10^i-1}\le\frac{100}{9}\\
&10^i\sum_{j=10^{p-2}}^{10^{p-1}-2}\ln(1+\frac{\frac{9\times10^{p-1}}{10(j+1)+d}}{(10j+d+1)10^i-10^{p-1}})\le10^{i}
\frac{\frac{10^{p}}{10^{p-1}}}{10^{p-1}10^i-10^{p-1}}\le\frac{100}{9\times10^{p-1}}\\
&10^{p-2+i}\ln(1+\frac{\frac{d\times10^{p-1}}{10^{p-1}+d}}{10^{p-1}10^i-10^{p-1}})\le10^{p-2+i}
\frac{\frac{d\times10^{p-1}}{10^{p-1}+d}}{10^{p-1}10^i-10^{p-1}}\le\frac{10^{p-1+i}}{10^{p-1}10^i-10^{p-1}}\le\frac{10}{9}\\
&10^{p-1+i}\ln(1+\frac{\frac{10^{p-1}(10-d-1)}{10^p}}{(10^p-10+d+1)10^i-10^{p-1}})\le
\frac{10^{p-1+i}}{10^{p-1}10^i-10^{p-1}}\le\frac{10}{9}\,.
\end{align*}
\endgroup
Thanks to $P_{(d,\phi_{(d,p)}(n),p)}$ upper bound and the above inequalities, the result follows.
\end{proof}

Let us denote by $\alpha_{(d,p)}$ the limit of $(P_{(d,\phi_{(d,p)}(n),p)})_{n\in\mathbb N\setminus\llbracket0,p-1\rrbracket}$. Here is a few values of 
$P_{(d,\phi_{(d,p)}(n),p)}$:

\begin{table}[ht]
\centering
\begin{tabular}{|c||c|c|c|c||c|}
\hline
\rowcolor{gray!40}$d$&$P_{(d,\phi_{(d,2)}(2),2)}$&$P_{(d,\phi_{(d,2)}(3),2)}$&$P_{(d,\phi_{(d,2)}(4),2)}$&$P_{(d,\phi_{(d,2)}(5),2)}$&$\alpha_{(d,2)}$\\\hline
$0$& $0.1330$ & $0.1144$ & $0.1123$ & $0.1121$ & $0.1121$  \\\hline
$1$& $0.1190$ & $0.1103$ & $0.1092$ & $0.1091$ & $0.1091$  \\\hline
$2$& $0.1107$ & $0.1068$ & $0.1063$ & $0.1062$ & $0.1062$  \\\hline
$3$& $0.1044$ & $0.1037$ & $0.1035$ & $0.1035$ & $0.1035$  \\\hline
$4$& $0.0991$ & $0.1007$ & $0.1009$ & $0.1009$ & $0.1009$ \\\hline
$5$& $0.0945$ & $0.0979$ & $0.0983$ & $0.0984$ & $0.0984$  \\\hline
$6$& $0.0903$ & $0.0953$ & $0.0958$ & $0.0959$ & $0.0959$  \\\hline
$7$& $0.0865$ & $0.0927$ & $0.0935$ & $0.0936$ & $0.0936$  \\\hline
$8$& $0.0829$ & $0.0902$ & $0.0912$ & $0.0913$ & $0.0913$  \\\hline
$9$& $0.0796$ & $0.0879$ & $0.0889$ & $0.0891$ & $0.0891$  \\\hline
\end{tabular}
\caption{Values of $P_{(d,\phi_{(d,2)}(n),2)}$ and $\alpha_{(d,2)}$, for $n\in\llbracket2,5\rrbracket$. These values are rounded to the nearest 
ten-thousandth.}
\label{tab5}
\end{table}

\begin{table}[ht]
\centering
\begin{tabular}{|c||c|c|c||c|}
\hline
\rowcolor{gray!40}$d$&$P_{(d,\phi_{(d,3)}(3),3)}$&$P_{(d,\phi_{(d,3)}(4),3)}$&$P_{(d,\phi_{(d,3)}(5),3)}$&$\alpha_{(d,3)}$\\\hline
$0$& $0.1045$ & $0.1015$ & $0.1012$ &  $0.1012$  \\\hline
$1$& $0.1028$ & $0.1011$ & $0.1009$ &  $0.1009$  \\\hline
$2$& $0.1017$ & $0.1008$ & $0.1007$ &  $0.1006$  \\\hline
$3$& $0.1008$ & $0.1004$ & $0.1004$ &  $0.1004$  \\\hline
$4$& $0.1000$ & $0.1001$ & $0.1001$ &  $0.1001$ \\\hline
$5$& $0.0993$ & $0.0998$ & $0.0999$ &  $0.0999$  \\\hline
$6$& $0.0986$ & $0.0995$ & $0.0996$ &  $0.0996$  \\\hline
$7$& $0.0980$ & $0.0992$ & $0.0993$ &  $0.0994$  \\\hline
$8$& $0.0974$ & $0.0989$ & $0.0991$ &  $0.0991$  \\\hline
$9$& $0.0968$ & $0.0986$ & $0.0988$ &  $0.0989$  \\\hline
\end{tabular}
\caption{Values of $P_{(d,\phi_{(d,3)}(n),3)}$ and $\alpha_{(d,3)}$, for $n\in\llbracket3,5\rrbracket$. These values are rounded to the nearest 
ten-thousandth.}
\label{tab6}
\end{table}

\section{Graphs of \texorpdfstring{$(P_{(d,n,p)})_{n\in\mathbb N\setminus\llbracket0,10^{p-1}-1\rrbracket}$}{}}

Let us plot graphs of sequences $(P_{(d,n,2)})_{n\in\mathbb N\setminus\llbracket0,10^{p-1}-1\rrbracket}$ for values of $n$ from $10$ to $1000$ 
(Figure \ref{fig3}). Then we plot graphs of $(P_{(d,n,3)})_{n\in\mathbb N\setminus\llbracket0,10^{p-1}-1\rrbracket}$, for 
$n\in\llbracket100,20000\rrbracket$ (Figure \ref{fig4}).

\begin{figure}[ht]
\begin{minipage}{0.48\linewidth}
\centering
\includegraphics[scale=0.45,clip=true]{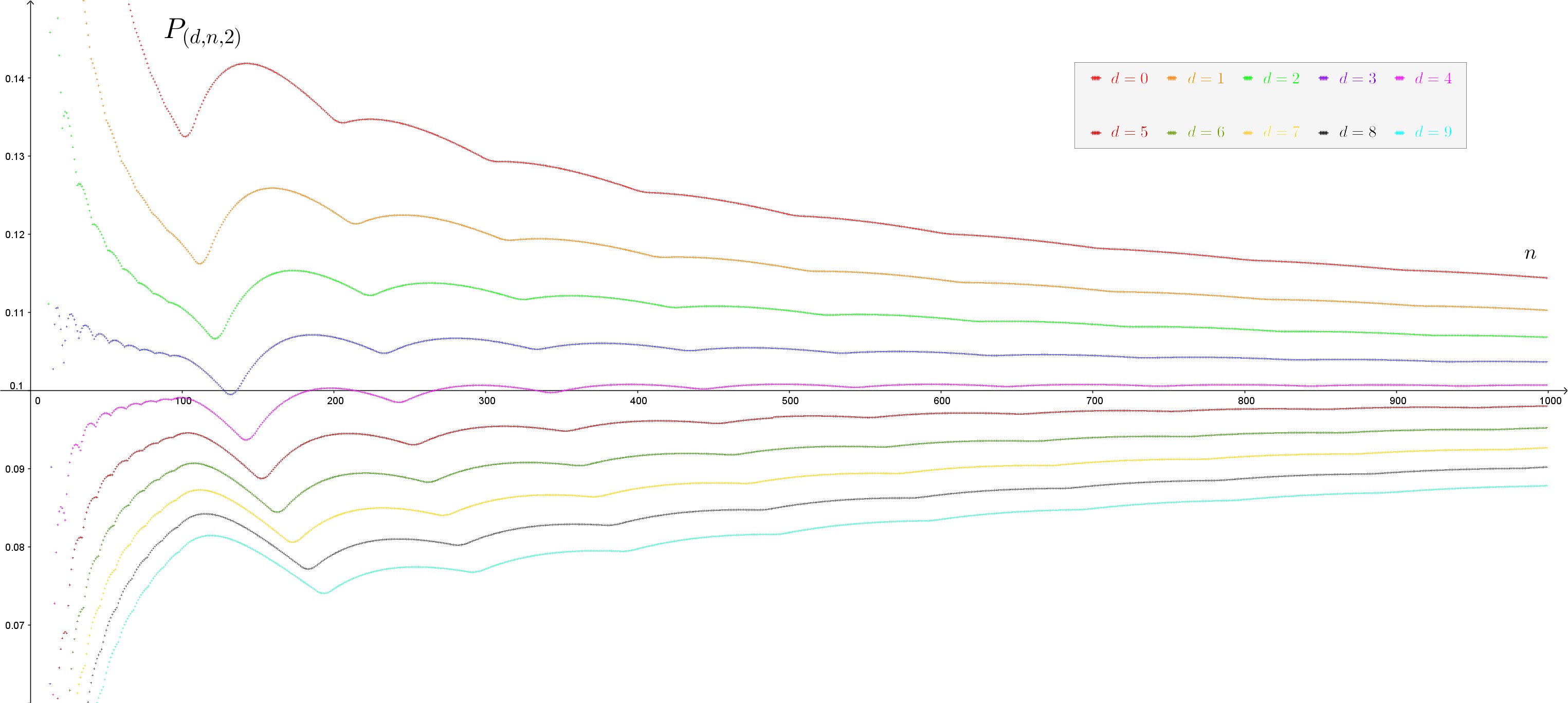}
\caption{For $d\in\llbracket0,9\rrbracket$, graphs of $(P_{(d,n,2)})_{n\in\mathbb N\setminus\llbracket0,10^{p-1}-1\rrbracket}$.}\label{fig3}
\vspace{0.695\baselineskip}

\end{minipage}\hfill
\begin{minipage}{0.48\linewidth}
\centering
\includegraphics[scale=0.55,clip=true]{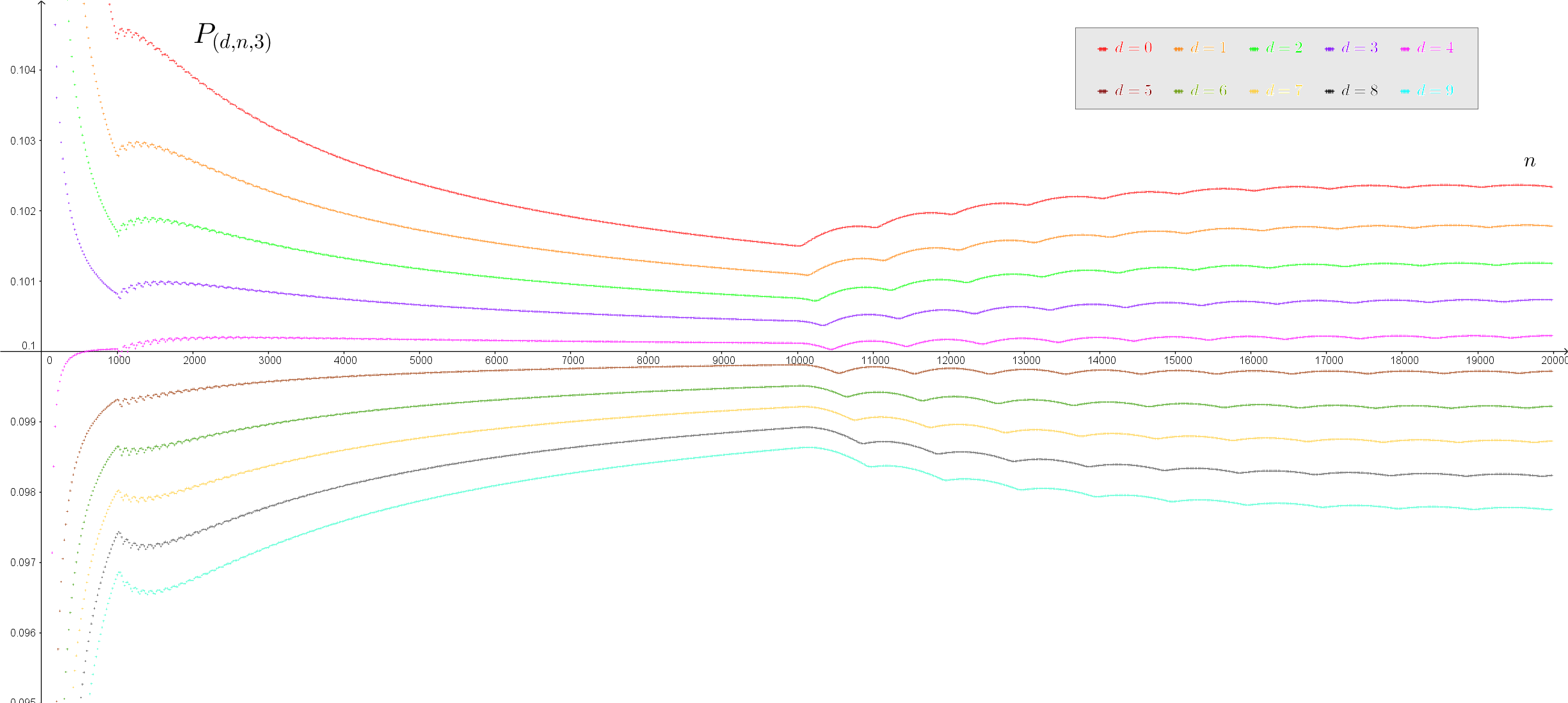}
\caption{For $d\in\llbracket0,9\rrbracket$, graphs of $(P_{(d,n,3)})_{n\in\mathbb N\setminus\llbracket0,10^{p-1}-1\rrbracket}$. Note that points have 
not been all represented.}\label{fig4}
\end{minipage}
\end{figure}

Let us plot two additional graphs of $P_{(d,n,2)}$ \textit{versus} $\log(n)$ and $P_{(d,n,3)}$ \textit{versus} $\log(n)$ for values of $n$ from $10$ to 
$2000000$:

\begin{figure}[ht]
\begin{minipage}{0.48\linewidth}
\centering
\includegraphics[scale=0.53,clip=true]{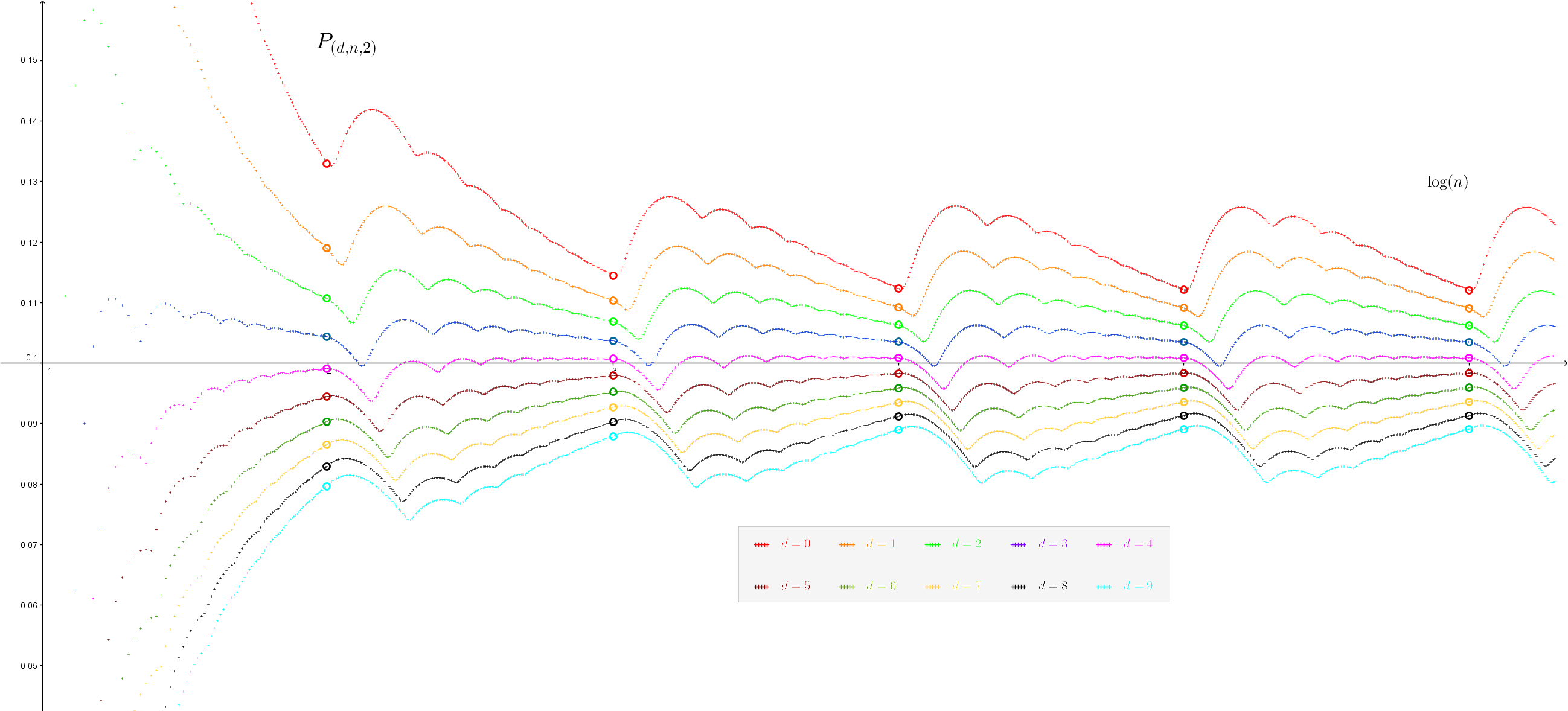}
\caption{For $d\in\llbracket0,9\rrbracket$, graphs of $P_{(d,n,2)}$ \textit{versus} $\log(n)$. Note that points have not been all ploted. 
The first five values of the above defined subsequence, for each $d$, being represented by bigger plots.}\label{fig5}
\vspace{0.695\baselineskip}

\end{minipage}\hfill
\begin{minipage}{0.48\linewidth}
\centering
\includegraphics[scale=0.64,clip=true]{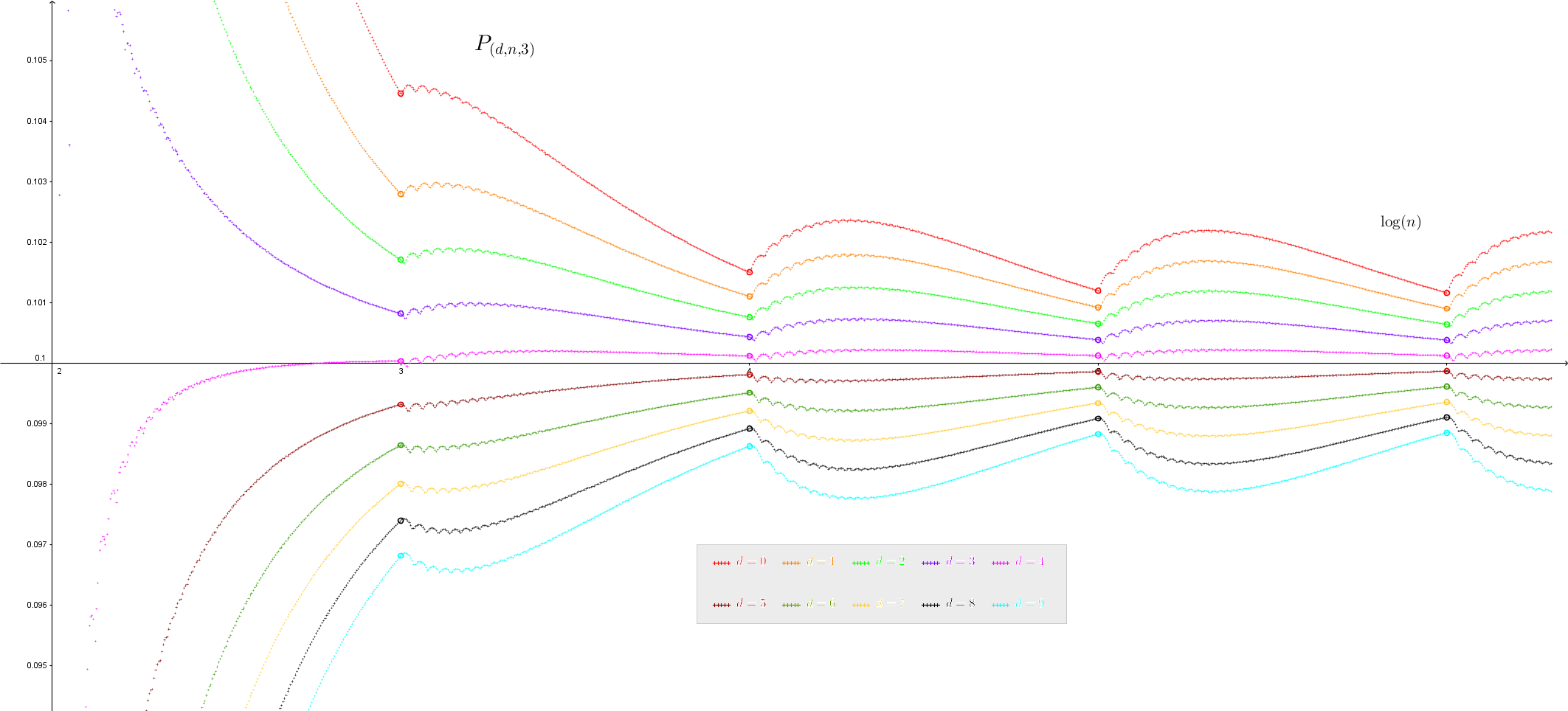}
\caption{For $d\in\llbracket0,9\rrbracket$, graphs of $P_{(d,n,3)}$ \textit{versus} $\log(n)$. Note that points have not been all ploted. 
The first four values of the above defined subsequence, for each $d$, being represented by bigger plots.}\label{fig6}
\end{minipage}
\end{figure}

Through Figures \ref{fig5} and \ref{fig6}, the proportion of each $d$ as leading digit, $d\in\llbracket0,9\rrbracket$, seems to fluctuate and 
consequently not follow Benford's Law. Each "pseudo cycle" seems to be composed of $9\times10^{p-2}$ short waves. Note that these observations were not 
obvious in view of Figures \ref{fig3} and \ref{fig4}.

We can also prove the following result:

\begin{proposition}
For all $n\in\mathbb N\setminus\llbracket0,10^{p-1}-1\rrbracket$ such that $n\ge10^{p-1}+9$ and for all $(a,b)\in\llbracket0,9\rrbracket^2$ such that 
$a<b$, we have:
\begin{align*}
P_{(a,n,p)}>P_{(b,n,p)}\,.
\end{align*}
\end{proposition}

The relative position of graphs of $P_{(d,n,p)}$, for $d\in\llbracket0,9\rrbracket$, can be observed on Figures \ref{fig3}, \ref{fig4} and \ref{fig5}.

\begin{proof}
$(a,b)\in\llbracket0,9\rrbracket^2$ such that $a<b$. For all $m\in\llbracket 10^{p-1},n\rrbracket$, let us denote by $\mathscr E_{(a,m)}$ 
the subset of $\mathbb N$ such that $\mathscr E_{(a,m)}=\{j\le m:\text{the }p^{\text{th}}\text{ digit of $j$ is }a\}$.

For all $e\in\mathscr E_{(b,m)}$, we consider $e'=e-(b-a)\times10^{dg-p}$ where $dg$ is the number of digits of the integer $e$. It is clear that 
$e'\in\mathscr E_{(a,m)}$. Thus we get: $\big|E_{(a,m)}\big|\ge\big|E_{(b,m)}\big|$. 

We also have $P_{(a,10^{p-1}+a,p)}=\frac{1}{a+1}>P_{(b,10^{p-1}+a,p)}=0$. The result follows.
\end{proof}

\begin{remark}
For $n\in\mathbb N\setminus\llbracket0,10^{p-1}-1\rrbracket$, we have, if $n<10^{p-1}+d$, $P_{(d,n,p)}=0$. Hence for all 
$n\in\mathbb N^\setminus\llbracket0,10^{p-1}-1\rrbracket$ and for all $(a,b)\in\llbracket0,9\rrbracket^2$ such that $a<b$, we have:
\begin{align*}
P_{(a,n,p)}\ge P_{(b,n,p)}\,.
\end{align*}
\end{remark}

Let us henceforth provide the following equality:

\begin{proposition}\label{pro2i}
\begin{align*}
P_{(d,n,p)}=\frac{1}{n+1-10^{p-1}}\Big(P_{(d,10^{k+p}-1,p)}\times(10^{k+p}-10^{p-1})+r_{(n,d,p)}\Big)\,,
\end{align*}
where:
\begin{center}
$k=\max\{i\in\mathbb N:10^{i+p}\le n\}\,$.
\end{center}
\end{proposition}

\begin{proof}
Results are directly derived from Proposition \ref{proi}.
\end{proof}

\section{Study of \texorpdfstring{$9\times10^{p-2}$}{} additional subsequences}

To definitively bring to light the fact that the sequence $(P_{(d,n,p)})_{n\in\mathbb N\setminus\llbracket0,10^{p-1}-1\rrbracket}$ does not converge, 
we will show that there exist additional subsequences that converge to limits different from those of 
$(P_{(d,\phi_{(d,p)}(n),p)})_{n\in\mathbb N\setminus\llbracket0,p-1\rrbracket}$.

For $i\in\llbracket10^{p-2},10^{p-1}-1\rrbracket$, let us in this way study the $9\times10^{p-2}$ 
subsequences $(P_{(d,\psi_{(d,p,i)}(n),p)})_{n\in\mathbb N\setminus\llbracket0,p-1\rrbracket}$ where $\psi_{(d,p,i)}$ is the function from 
$\mathbb N\setminus\llbracket0,p-1\rrbracket$ to $\mathbb N$ that maps $n$ to $(10i+(d+1))10^{n-p+1}-1$. We get the below result:
\begin{proposition}$i\in\llbracket10^{p-2},10^{p-1}-1\rrbracket$.\\
The subsequence $(P_{(d,\psi_{(d,p,i)}(n),p)})_{n\in\mathbb N\setminus\llbracket0,p-1\rrbracket}$ converges to:
\begingroup\scriptsize
\begin{align*}
\frac{\alpha_{(d,p)}10^{p-1}+i+1-10^{p-2}-k_{(d,p,i)}d-9l_{(d,p,i)}+m_{(d,p,i)}+n_{(d,p,i)}+10^{p-2}
\ln(\frac{10^{p-1}+d}{10^{p-1}})}{10i+d+1}\,,
\end{align*}
\endgroup
where:
\[\left \{
\begin{array}{l @{=} l}
    k_{(d,p,i)} & \sum_{j=10^{p-2}}^{i}\ln(\frac{10j+(d+1)}{10j+d}) \\
    l_{(d,p,i)} & \sum_{j=10^{p-2}}^{i}j\ln(\frac{10j+(d+1)}{10j+d}) \\
    m_{(d,p,i)} & \sum_{j=10^{p-2}}^{i-1}\ln(\frac{10(j+1)+d}{10j+(d+1)}) \\
    n_{(d,p,i)} & \sum_{j=10^{p-2}}^{i-1}j\ln(\frac{10(j+1)+d}{10j+(d+1)})\,.
\end{array}
\right.\]
\end{proposition}

\begin{proof}$i\in\llbracket10^{p-2},10^{p-1}-1\rrbracket$.
Thanks to Proposition \ref{pro2i}, we have, for $n\in\mathbb N\setminus\llbracket0,p-1\rrbracket$:
\begin{align*}
P_{(d,\psi_{(d,p,i)}(n),p)}&=\frac{1}{\big(10i+(d+1)\big)10^{n-p+1}-10^{p-1}}\Big(P_{(d,10^n-1,p)}\times(10^n-10^{p-1})\\
&\quad+r_{(\psi_{(d,p,i)}(n),d,p)}\Big)\,.
\end{align*}
The first term of $r_{(\psi_{(d,p,i)}(n),d,p)}$ can be simplified as follows:
\begingroup\small
\begin{align*}
&\sum_{j=10^{p-2}}^i\sum_{b=(10j+d)10^{n-p+1}}^{(10j+(d+1))10^{n-p+1}-1}\Big(1-\frac{(9j+d)10^{n-p+1}+10^{p-2}-10^{p-1}}{b+1-10^{p-1}}\Big)\\
&\quad=10^{n-p+1}(i-10^{p-2}+1)-\sum_{j=10^{p-2}}^i\big((9j+d)10^{n-p+1}+10^{p-2}-10^{p-1}\big)\\
&\quad\sum_{b=(10j+d)10^{n-p+1}}^{(10j+(d+1))10^{n-p+1}-1}\frac{1}{b+1-10^{p-1}}\\
&\underset{\substack{n \to +\infty}}{\sim}10^{n-p+1}(i-10^{p-2}+1)-\sum_{j=10^{p-2}}^i(9j+d)10^{n-p+1}\ln(\frac{10j+(d+1)}{10j+d})\,,
\end{align*}
\endgroup
thanks to inequalities \ref{inelog}.\\
The second term of $r_{(\psi_{(d,p,i)}(n),d,p)}$ can be simplified as follows:
\begin{align*}
&\sum_{j=10^{p-2}-1}^{i-1}\sum_{a=\max(10^{n},(10j+(d+1))10^{n-p+1})}^{(10(j+1)+d)10^{n-p+1}-1}\frac{10^{n-p+1}(j+1)-10^{p-2}}{a+1-10^{p-1}}\\
&\quad=\big(10^{n-p+1}10^{p-2}-10^{p-2}\big)\sum_{a=10^{n}}^{(10^{p-1}+d)10^{n-p+1}-1}\frac{1}{a+1-10^{p-1}}\\
&\quad+\big(10^{n-p+1}(j+1)-10^{p-2}\big)\sum_{j=10^{p-2}}^{i-1}\sum_{a=(10j+(d+1))10^{n-p+1}}^{(10(j+1)+d)10^{n-p+1}-1}\frac{1}{a+1-10^{p-1}}\\
&\underset{\substack{n \to +\infty}}{\sim}10^{n-1}\ln(\frac{10^{p-1}+d}{10^{p-1}})+\sum_{j=10^{p-2}}^{i-1}10^{n-p+1}(j+1)
\ln(\frac{10(j+1)+d}{10j+(d+1)})\,,
\end{align*}
thanks to inequalities \ref{inelog}.\\
Knowing that $P_{(d,10^n-1,p)}\underset{\substack{n \to +\infty}}{\sim}\alpha_{(d,p)}$ (see Proposition \ref{sub1i}), the result follows.
\end{proof}

Let us denote by $\alpha_{(d,p,i)}$ the limit of $(P_{(d,\psi_{(d,p,i)}(n),p)})_{n\in\mathbb N\setminus\llbracket0,p-1\rrbracket}$. Here is a few values 
of $P_{(d,\psi_{(d,p,i)}(n),p)}$:

\begin{table}[ht]
\centering
\begingroup\scriptsize
\begin{tabular}{|c||c|c|c|c||c|}
\hline
\rowcolor{gray!40}$d$&$P_{(d,\psi_{(d,2,7)}(2),2)}$&$P_{(d,\psi_{(d,2,7)}(3),2)}$&$P_{(d,\psi_{(d,2,7)}(4),2)}$&$P_{(d,\psi_{(d,2,7)}(5),2)}$
&$\alpha_{(d,2,7)}$\\\hline
$0$& $0.1182$ & $0.1152$ & $0.1148$ & $0.1148$ & $0.1148$  \\\hline
$1$& $0.1127$ & $0.1111$ & $0.1109$ & $0.1109$ & $0.1109$  \\\hline
$2$& $0.1082$ & $0.1074$ & $0.1073$ & $0.1073$ & $0.1073$  \\\hline
$3$& $0.1042$ & $0.1040$ & $0.1040$ & $0.1040$ & $0.1039$  \\\hline
$4$& $0.1006$ & $0.1008$ & $0.1008$ & $0.1008$ & $0.1008$ \\\hline
$5$& $0.0973$ & $0.0978$ & $0.0979$ & $0.0979$ & $0.0979$  \\\hline
$6$& $0.0942$ & $0.0950$ & $0.0951$ & $0.0951$ & $0.0951$  \\\hline
$7$& $0.0913$ & $0.0923$ & $0.0924$ & $0.0925$ & $0.0925$  \\\hline
$8$& $0.0886$ & $0.0898$ & $0.0899$ & $0.0900$ & $0.0900$  \\\hline
$9$& $0.0860$ & $0.0874$ & $0.0876$ & $0.0876$ & $0.0876$  \\\hline
\end{tabular}
\endgroup
\caption{Values of $P_{(d,\psi_{(d,2,7)}(n),2)}$ and $\alpha_{(d,2,7)}$, for $n\in\llbracket2,5\rrbracket$ and $i=7$. These values are rounded to the 
nearest ten-thousandth.}
\label{tab7}
\end{table}

\begin{table}[ht]
\centering
\begingroup\scriptsize
\begin{tabular}{|c||c|c|c||c|}
\hline
\rowcolor{gray!40}$d$&$P_{(d,\psi_{(d,3,23)}(3),3)}$&$P_{(d,\psi_{(d,3,23)}(4),3)}$&$P_{(d,\psi_{(d,3,23)}(5),3)}$&$\alpha_{(d,3,23)}$\\\hline
$0$& $0.1037$ & $0.1023$ & $0.1022$ &  $0.1021$ \\\hline
$1$& $0.1026$ & $0.1018$ & $0.1017$ &  $0.1017$ \\\hline
$2$& $0.1017$ & $0.1012$ & $0.1012$ &  $0.1012$ \\\hline
$3$& $0.1009$ & $0.1007$ & $0.1007$ &  $0.1007$ \\\hline
$4$& $0.1007$ & $0.1002$ & $0.1002$ &  $0.1002$ \\\hline
$5$& $0.0995$ & $0.0997$ & $0.0997$ &  $0.0997$ \\\hline
$6$& $0.0988$ & $0.0992$ & $0.0993$ &  $0.0993$ \\\hline
$7$& $0.0982$ & $0.0987$ & $0.0988$ &  $0.0988$ \\\hline
$8$& $0.0976$ & $0.0983$ & $0.0983$ &  $0.0983$ \\\hline
$9$& $0.0969$ & $0.0978$ & $0.0979$ &  $0.0979$ \\\hline
\end{tabular}
\endgroup
\caption{Values of $P_{(d,\psi_{(d,3,23)}(n),3)}$ and $\alpha_{(d,3,23)}$, for $n\in\llbracket3,5\rrbracket$ and $i=23$. These values are rounded to the 
nearest ten-thousandth.}
\label{tab8}
\end{table}

As a result, the sequence $(P_{(d,n,p)})_{n\in\mathbb N\setminus\llbracket0,10^{p-1}-1\rrbracket}$ does not converge. The $9\times10^{p-2}$ convergent 
subsequences confirm the remarks raised by Figures \ref{fig5} and \ref{fig6} about the existence of "pseudo cycles" in the graph of 
$(P_{(d,n,p)})_{n\in\mathbb N\setminus\llbracket0,10^{p-1}-1\rrbracket}$.

\subsection{Central values}

From Figures \ref{fig5} and \ref{fig6}, we notice that there exist fluctuations in the graph of $(P_{(d,n,p)})_
{n\in\mathbb N\setminus\llbracket0,10^{p-1}-1\rrbracket}$. We define $C_{(d,p)}$ as follows:

\begin{definition}
\begin{align*}
C_{(d,p)}=\frac{1}{9\times10^{p-2}}\sum_{i=10^{p-2}}^{10^{p-1}-1}\alpha_{(d,p,i)}\,.
\end{align*}
\end{definition}

Figure \ref{fig7} below shows the different values of $\alpha_{(0,2,i)}$, for $i\in\llbracket1,9\rrbracket$ and also the values of $P_{(0,n,2)}$ 
\textit{versus} $\log(n)$ for $n\in\llbracket10,2000000\rrbracket$:

\begin{figure}[ht]
\centering
\includegraphics[scale=0.75,clip=true]{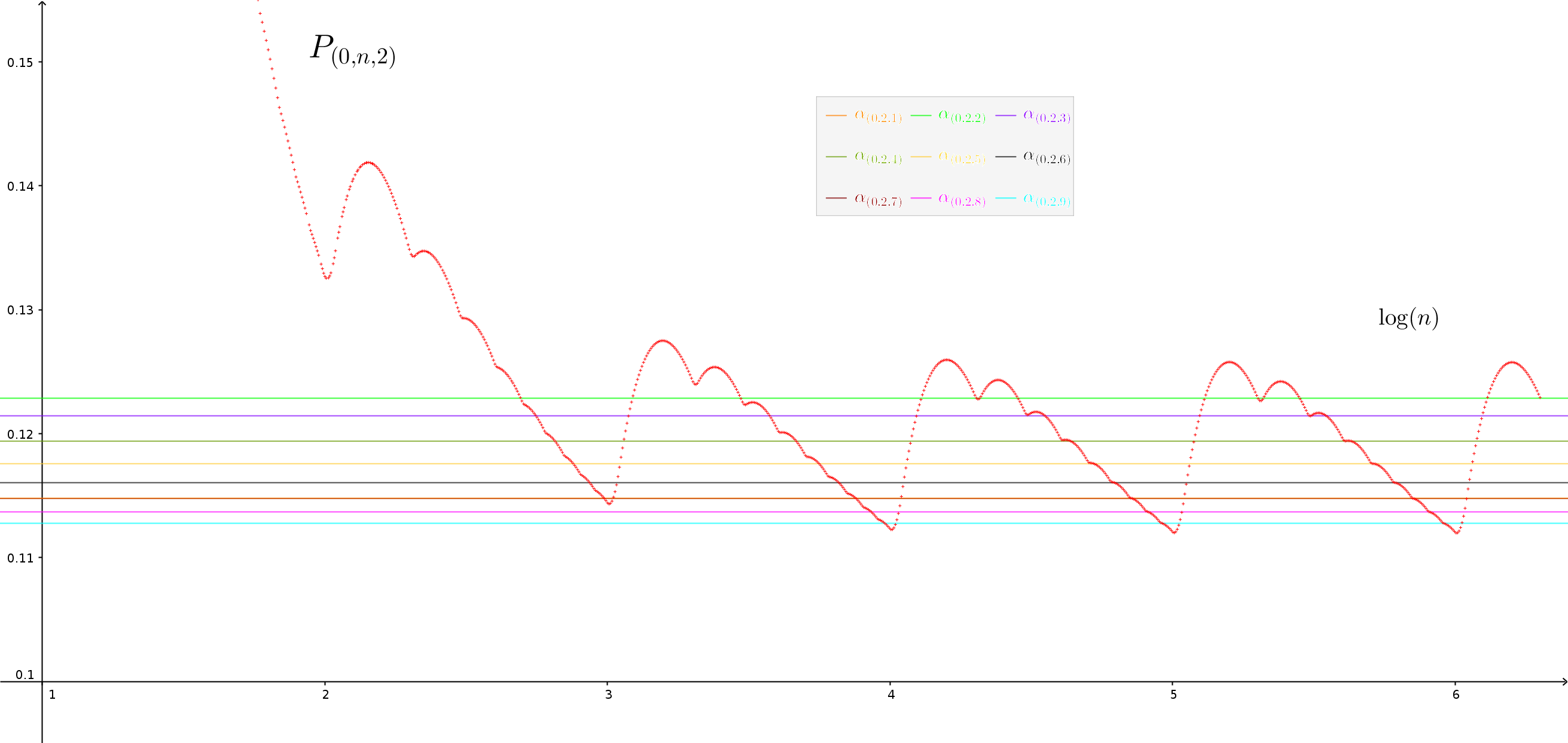}
\caption{Graph of $P_{(0,n,2)}$ \textit{versus} $\log(n)$. Note that points have not been all represented. 
Lines whose equation is $y=\alpha_{(0,2,i)}$, for $i\in\llbracket1,9\rrbracket$, have also been ploted. Note that those of equations 
$y=\alpha_{(0,2,1)}$ and $y=\alpha_{(0,2,7)}$ are almost coincident. We have $C_{(0,2)}\approx0.1170$.}\label{fig7}
\vspace{0.695\baselineskip}
\end{figure}

These means values are very close to the theoric value highlighted in \cite{hit} as can be seen in below tables (Tables \ref{tab9} and \ref{tab10}, 
where $p=2$ and $p=3$, respectively). According to Hill (\cite{hip}), it is absolutely normal. 

\begin{table}[ht]
\centering
\begin{tabular}{|c||c|c|}
\hline
\rowcolor{gray!40}$d$&$C_{(d,2)}$&$\sum_{j=1}^{9}\log(1+\frac{1}{10j+d})$\\\hline
$0$& $0.1170$ & $0.1197$  \\\hline
$1$& $0.1122$ & $0.1139$  \\\hline
$2$& $0.1079$ & $0.1088$  \\\hline
$3$& $0.1039$ & $0.1043$  \\\hline
$4$& $0.1001$ & $0.1003$  \\\hline
$5$& $0.0967$ & $0.0967$  \\\hline
$6$& $0.0935$ & $0.0934$  \\\hline
$7$& $0.0905$ & $0.0904$  \\\hline
$8$& $0.0878$ & $0.0876$  \\\hline
$9$& $0.0851$ & $0.0850$  \\\hline
\end{tabular}
\caption{Values of $C_{(d,p)}$ and probabilities associated to the second digit (\cite{hit}), for $p=2$. These values are 
rounded to the nearest thousandth.}
\label{tab9}
\end{table}

We furthermore note, thanks to Table \ref{tab9}, that $C_{(0,2)}$ slightly underestimates $\sum_{j=1}^{9}\log(1+\frac{1}{10j})$ as can be infered from 
Figure \ref{fig7}.

\begin{table}[ht]
\centering
\begin{tabular}{|c||c|c|}
\hline
\rowcolor{gray!40}$d$&$C_{(d,3)}$&$\sum_{j=10}^{99}\log(1+\frac{1}{10j+d})$\\\hline
$0$& $0.1016$ & $0.1018$  \\\hline
$1$& $0.1013$ & $0.1014$  \\\hline
$2$& $0.1009$ & $0.1010$  \\\hline
$3$& $0.1005$ & $0.1006$  \\\hline
$4$& $0.1002$ & $0.1002$  \\\hline
$5$& $0.0998$ & $0.0998$  \\\hline
$6$& $0.0994$ & $0.0994$  \\\hline
$7$& $0.0991$ & $0.0990$  \\\hline
$8$& $0.0987$ & $0.0986$  \\\hline
$9$& $0.0984$ & $0.0983$  \\\hline
\end{tabular}
\caption{Values of $C_{(d,p)}$ and probabilities associated to the third digit (\cite{hit}). These values are 
rounded to the nearest thousandth.}
\label{tab10}
\end{table}

\section*{Conclusion}

To conclude, through our model, we have seen that the proportion of $d$ as $p^{\text{th}}$ digit, $d\in\llbracket0,9\rrbracket$, in certain naturally 
occurring collections of data is more likely to follow a law whose probability distribution is $(d,P_{(d,n,p)})_{d\in\llbracket0,9\rrbracket}$, where 
$n$ is the smaller integer upper bound of the physical, biological or economical quantities considered, rather than the generalized Benford's Law. 
Knowing beforehand the value of the upper bound $n$ can be a way to find a better adjusted law than Benford's one. 

The results of the article would have been the same in terms of fluctuations of the proportion of $d\in\llbracket0,9\rrbracket$ as $p^{\text{th}}$ digit, of 
limits of subsequences, or of results on central values, if our discrete uniform distributions uniformly randomly selected were lower bounded by a 
positive integer different from $10^{p-1}$: first terms in proportion formulas become rapidly negligible. Through our model we understand that the 
predominance of $0$ as $p^{\text{th}}$ digit (followed by those of $1$ and so on) is all but surprising in experimental data: it is only due to the fact 
that, in the lexicographical order, $0$ appears before $1$, $1$ appears before $2$, \textit{etc.}

However the limits of our model rest on the assumption that the random variables used to obtain our data are not the same and follow discrete uniform 
distributions that are uniformly randomly selected. In certain naturally occurring collections of data it cannot conceivably be justified. Studying the 
cases where the random variables follow other distributions (and not necessarily randomly selected) sketch some avenues for future research on the 
subject.

\bibliography{bib}

\section*{Appendix: Python script}

Using Proposition \ref{proi}, we can determine the terms of $(P_{(d,n,p)})_{n\in\mathbb N\setminus\llbracket0,10^{p-1}-1\rrbracket}$, for 
$d\in\llbracket0,9\rrbracket$. To this end, we have created a script with the Python programming language (Python Software Foundation, Python Language 
Reference, version $3.4.$ available at \url{http://www.python.org}, see \cite{pyt}). The implemented function \textit{expvalProp} has three parameters: 
the rank $n$ of the wanted term of the sequence, the position $p$ of the considered digit and the value $d$ of this digit. Here is the used algorithm:
\vspace{0.5\baselineskip}

\begingroup\scriptsize
{\setlength{\parindent}{1.5em}\textit{def expvalProp(n,d,p):}}

{\setlength{\parindent}{3em}\textit{k=-1;}}

{\setlength{\parindent}{3em}\textit{while(10**(k+p+1)<=n):}}

{\setlength{\parindent}{4.5em}\textit{k=k+1}}

{\setlength{\parindent}{3em}\textit{l=math.floor((n-(10**(p-1)+d)*10**(k+1))/10**(k+2))+10**(p-2);S=0;T=0;}}

{\setlength{\parindent}{3em}\textit{if (k!=-1):}}

{\setlength{\parindent}{4.5em}\textit{for i in range(0,k+1):}}

{\setlength{\parindent}{6em}\textit{for j in range(10**(p-2),10**(p-1)):}}

{\setlength{\parindent}{7.5em}\textit{for b in range((10*j+d)*10**i,(10*j+(d+1))*10**i):}}

{\setlength{\parindent}{9em}\textit{T=T+(b-((9*j+d)*10**i+10**(p-2)-1))/(b+1-10**(p-1))}}

{\setlength{\parindent}{6em}\textit{for j in range(10**(p-2)-1,10**(p-1)):}}

{\setlength{\parindent}{7.5em}\textit{for a in range(max(10**(p+i-1),(10*j+(d+1))*10**i),min(10**(p+i),(10*}}

{\setlength{\parindent}{33em}\textit{(j+1)+d)*10**i)):}}

{\setlength{\parindent}{9em}\textit{S=S+((j+1)*10**i-10**(p-2))/(a+1-10**(p-1))}}

{\setlength{\parindent}{3em}\textit{if ((math.floor(n/10**(k+1))-10*math.floor(n/10**(k+2)))==d):}}

{\setlength{\parindent}{4.5em}\textit{for j in range(10**(p-2),l+1):}}

{\setlength{\parindent}{6em}\textit{for b in range((10*j+d)*10**(k+1),min(n,(10*j+(d+1))*10**(k+1)-1)+1):}}

{\setlength{\parindent}{7.5em}\textit{T=T+(b-((9*j+d)*10**(k+1)+10**(p-2)-1))/(b+1-10**(p-1))}}

{\setlength{\parindent}{4.5em}\textit{for j in range(10**(p-2)-1,l):}}

{\setlength{\parindent}{6em}\textit{for a in range(max(10**(p+k),(10*j+(d+1))*10**(k+1)),(10*(j+1)+d)*10**}}

{\setlength{\parindent}{38em}\textit{(k+1)):}}

{\setlength{\parindent}{7.5em}\textit{S=S+((j+1)*10**(k+1)-10**(p-2))/(a+1-10**(p-1))}}

{\setlength{\parindent}{3em}\textit{else:}}

{\setlength{\parindent}{4.5em}\textit{for j in range(10**(p-2),l+1):}}

{\setlength{\parindent}{6em}\textit{for b in range((10*j+d)*10**(k+1),(10*j+(d+1))*10**(k+1)):}}

{\setlength{\parindent}{7.5em}\textit{T=T+(b-((9*j+d)*10**(k+1)+10**(p-2)-1))/(b+1-10**(p-1))}}

{\setlength{\parindent}{4.5em}\textit{for j in range(10**(p-2)-1,l+1):}}

{\setlength{\parindent}{6em}\textit{for a in range(max(10**(p+k),(10*j+(d+1))*10**(k+1)),min(n,(10*(j+1)+d)}}

{\setlength{\parindent}{33em}\textit{*10**(k+1)-1)+1):}}

{\setlength{\parindent}{7.5em}\textit{S=S+((j+1)*10**(k+1)-10**(p-2))/(a+1-10**(p-1))}}

{\setlength{\parindent}{3em}\textit{return((S+T)/(n+1-10**(p-1)))}}

\endgroup
\end{document}